\documentclass[twocolumn,10pt]{asme2ej}

\usepackage{amsmath,amssymb,amsfonts,dsfont}  
\usepackage{graphicx}                         
\usepackage[usenames,dvipsnames]{color}       
\usepackage[font=small]{caption}              
\usepackage[normalem]{ulem}	                  

\usepackage{algorithm,algorithmicx,listings}	  
\usepackage[noend]{algpseudocode}			          

\usepackage[pagebackref=false,%
            breaklinks=true,%
            letterpaper=true,%
            colorlinks,%
            bookmarks=true,%
            citecolor=Black,%
            urlcolor=Blue]{hyperref}

\def\argmax{\mathop{\arg\max}\limits}	%
\DeclareMathOperator{\tr}{tr}
\newcommand{\indicator}{\mathds{1}}			
\usepackage{extarrows}
\newcommand{\longeq}[2]{\xlongequal[\!#2\!]{\!#1\!}}
\newcommand{\longineq}[3]{\overset{#1}{\underset{#2}{#3}}}

\newtheorem{theorem}{Theorem}

\newtheorem{problem}{Problem}[section]
\newtheorem*{problem*}{Problem}

\title{Distributed Algorithms for Stochastic Source Seeking with Mobile Robot Networks: Technical Report%
\thanks{This work has been submitted to the ASME for possible publication. Copyright may be transferred without notice, after which this version may no longer be accessible.}
}

\author{Nikolay A. Atanasov%
\thanks{Address all correspondence to this author.\newline This work was supported by ONR-HUNT grant N00014-08-1-0696 and by TerraSwarm, one of six centers of STARnet, a Semiconductor Research Corporation program sponsored by MARCO and DARPA.}
    \affiliation{
	Electrical and Systems Engineering Dept.\\
	University of Pennsylvania\\
	Philadelphia, PA 19104\\
    Email: atanasov@seas.upenn.edu
    }
}

\author{Jerome Le Ny
    \affiliation{
  Electrical Engineering Dept.\\
	\'Ecole Polytechnique de Montr\'eal\\
	Montr\'eal, QC H3T-1J4, Canada\\
        Email: jerome.le-ny@polymtl.ca 
    }
}

\author{George J. Pappas
    \affiliation{
  Electrical and Systems Engineering Dept.\\
  University of Pennsylvania, Philadelphia, PA 19104\\
    Email: pappasg@seas.upenn.edu
    }
}

\begin{document}

\maketitle    

\begin{abstract}
\textit{
Autonomous robot networks are an effective tool for monitoring large-scale environmental fields. This paper proposes distributed control strategies for localizing the source of a noisy signal, which could represent a physical quantity of interest such as magnetic force, heat, radio signal, or chemical concentration. We develop algorithms specific to two scenarios: one in which the sensors have a precise model of the signal formation process and one in which a signal model is not available. In the model-free scenario, a team of sensors is used to follow a stochastic gradient of the signal field. Our approach is distributed, robust to deformations in the group geometry, does not necessitate global localization, and is guaranteed to lead the sensors to a neighborhood of a local maximum of the field. In the model-based scenario, the sensors follow the stochastic gradient of the mutual information between their expected measurements and the location of the source in a distributed manner. The performance is demonstrated in simulation using a robot sensor network to localize the source of a wireless radio signal.
}
\end{abstract}


\section{Introduction}
\label{sec:intro}
The ability to detect the source of a signal is a fundamental problem in nature. At a microscopic level, some bacteria are able to find chemical, light, and magnetic sources \cite{Lux04_chemotaxis, Frankel97_magnetotaxis}. At a macroscopic level similar behavior can be observed in predators who seek a food source using their sense of smell. Reproducing this behavior in mobile robots can be used to perform complex missions such as environmental monitoring \cite{Ogren04_multirobot, Sukhatme07_monitoring}, intelligence, surveillance, and reconnaissance \cite{Ribsky00_isr}, and search and rescue operations \cite{Kumar04_firstResponse}.

In this paper, we discuss how to control a team of mobile robotic sensors with the goal of locating the source of a noisy signal, which represents a physical quantity of interest such as magnetic force, heat, radio signal, or chemical concentration. We distinguish between two cases: \textit{model-free} and \textit{model-based}. The first scenario supposes that the sensors receive measurements without knowledge about the signal formation process. This is relevant when the signal is difficult to model or the environment is unknown a priori. On-line modeling of the signal might not be feasible either because it requires time and computational resources, which are limited on small platforms and in time-critical missions. The second scenario supposes that the sensors have an accurate signal model which can be exploited to localize the source, potentially faster and with better accuracy.


Our model-free source-seeking approach consists in climbing the gradient of the signal field by using a stochastic approximation technique to deal with the underlying noise. Our strategy is robust to deformations in the geometry of the sensor network and can be applied to sensors with limited computational resources and no global localization capabilities. Recent work developing model-free source-seeking using a sensor formation to ascend the gradient of the signal field includes \cite{Ogren04_multirobot, Zhang11_sourceSeeking, Guo_ICRA12, Arranz_ECC13}. \"Ogren et al. \cite{Ogren04_multirobot} use artificial potentials to decouple the formation stabilization from the gradient ascent. Centralized least-squares are used to estimate the signal gradient.
A distributed approach for exploring a scalar field using a cooperative Kalman filter is presented in \cite{FZhang_TAC10}. The authors design control laws to achieve a formation, which minimizes the estimation error.
Similarly, in \cite{Arranz_ECC13} a circular formation is used to estimate the signal gradient in a distributed manner based on a Newton-Raphson consensus method. A drawback of these works is the assumption that the sensor formation is maintained perfectly throughout the execution of the algorithm, which is hardly possible in a real environment. In this paper, imperfect formations are explicitly handled by re-computing the correct weights necessary to combine the sensor observations at every measurement location.
Choi et al. \cite{Choi_automatica09, Jadaliha_JDSMC12} present a general distributed learning and control approach for sensor networks and apply it to source seeking. The sensed signal is modeled by a network of radial basis functions and recursive least squares are used to obtain the model parameters. The convergence properties of the combined motion control and parameter estimation dynamics are analyzed.
Instead of a sensor network, a \textit{single} vehicle may travel to several sensing locations in order to collect the same measurements \cite{Azuma_TAC12,Zhang07_sourceSeeking, Liu10_sourceSeeking, Stankovic10_extremumSeeking, Ghods11_sourceSeeking}. While costly maneuvers are required to climb the gradient effectively, in our previous work \cite{Atanasov_SourceSeeking_ICRA12} we discussed parameter choices, which enable good performance.

In the model-based scenario, we choose the next configuration for the sensing team by maximizing the mutual information (MI) between the signal source estimate and the expected measurements. Even if all pose and measurement information is available at a central location, evaluating the MI utility function is computationally demanding. Charrow et al. \cite{Charrow_RSS13} focus on approximating MI when the sensed signal is Gaussian and the sensors use a particle filter to estimate the source location. Hoffman et al. \cite{Hoffman_TAC10} compute the expectation over the measurements only for pairs of sensors, thus decreasing the dimension of the required integration. Instead of MI, in this work we approximate the MI gradient. Related work which uses the MI gradient includes \cite{Dames_CDC12}, in which the dimension of the MI gradient is reduced by integrating over binary sensor measurements and only for sensors whose fields of view overlap. A fully distributed approach based on belief consensus is proposed in \cite{Julian_IJRR12}. This paper is also related to consensus control, which seeks agreement in the states of multi-agent dynamical systems. Recent results \cite{Yin_SIAM13,Yu_Auto11} address switching topologies, non-trivial delays, and asynchronous estimation but with the main difference that the sensors agree on their own states, while in this work they need to agree on the exogenous state of the source. 



\textbf{Contributions}: We develop a \textit{distributed} approach for stochastic source seeking using a mobile sensor network, which \textit{does not rely on a model of the signal field and global localization}. Our method uses a finite difference scheme to estimate the signal gradient correctly, \textit{even when the sensor formation is not maintained well}. In the model-based case, we show that a stochastic approximation to the MI gradient \textit{using only a few predicted signal measurements} is enough to provide good control performance. This is in contrast with existing work, which insists on improving the quality of the gradient estimate as much as possible. 


The rest of the paper is organized as follows. In Sec. \ref{sec:prob_form} we describe the considered source-seeking scenarios precisely. Our model-free and model-based approaches are discussed in detail in Sec. \ref{sec:model_free} and Sec. \ref{sec:model_based}, respectively, assuming all-to-all communication among the sensors. Distributed versions are presented and analyzed in Sec. \ref{sec:dist_alg}. Finally, in Sec. \ref{sec:applications} we present an application to wireless radio source localization and compare the performance of the two methods.

\section{Problem Formulation}
\label{sec:prob_form}
Consider a team of $n$ sensing robots with states $\{x_{1,t},\ldots,x_{n,t}\} \subset \mathcal{X} \cong \mathbb{R}^{d_x}$ at time $t$. The states are typically comprised of pose and velocity information but might include other operational parameters too. At a high-level planning stage we suppose that the vehicles have discrete single-integrator dynamics $x_{i,t+1} = x_{i,t} + u_{i,t}$, where $u_{i,t} \in \mathcal{U}$ is the control input to sensor $i$. The task is to localize a \textit{static} signal source, whose unknown state is $y \in \mathcal{Y} \cong \mathbb{R}^{d_y}$. The state captures the source position and other observable properties of interest. At time $t$ each sensor $i$ has access to a noisy measurement $z_{i,t} \in \mathcal{Z} \cong \mathbb{R}^{d_z}$ of the signal generated by $y$:
\begin{equation}
\label{eq:obs_model}
z_{i,t} = h(x_{i,t},y) + v_{i,t},
\end{equation}
where $v_{i,t}$ is the measurement noise, whose values are independent at any pair of times and among sensors. The noise depends on the states of the sensor and the source, i.e. $v_{i,t}(x_{i,t},y)$, but to simplify notation we do not make it explicit. We assume that the noise is zero-mean and has a finite second moment, i.e. $\mathbb{E} v_{i,t} = 0,\; \forall i,t,x_{i,t}$ and $\tr\bigl(\mathbb{E}[ v_{i,t}v_{i,t}^T]\bigr) < \infty$. In the reminder, we use the notation $x_t := \begin{bmatrix} x_{1,t}^T,\ldots,x_{n,t}^T \end{bmatrix}^T\mkern-12mu,\; u_t := \begin{bmatrix} u_{1,t}^T,\ldots,u_{n,t}^T \end{bmatrix}^T\mkern-12mu, \;\; z_t := \begin{bmatrix} z_{1,t}^T,\ldots,z_{n,t}^T \end{bmatrix}^T\mkern-12mu,\;$ and $v_t := \begin{bmatrix} v_{1,t}^T,\ldots,v_{n,t}^T\end{bmatrix}^T\mkern-12mu.$

In the model-free scenario, the sensors simply receive measurements without knowing the signal model $h(\cdot,\cdot)$. We suppose that the team adopts some arbitrary formation, with center of mass $m_t := \sum_{i=1}^n x_{i,t}/n$ at time $t$, which can be enforced using potential fields \cite{Ogren04_multirobot} or convex optimization \cite{Derenick_TRO09}. The sensors use the centroid $m_t$ as the estimate of the source state $y$ at time $t$ and try to lead it towards the true source location, based on the received measurements. Let $f:\mathcal{X} \rightarrow \mathcal{Y}$ be a known transformation, which maps the team centroid to a source estimate. For example, if the robot state space captures both position and orientation, e.g. $\mathcal{X} = SE(2)$, but we are interested only in position estimates for the source, e.g. $\mathcal{Y} = \mathbb{R}^2$, then $f$ would be the projection which extracts the position components from the centroid $m_t \in \mathcal{X}$. We consider the following problem.
\begin{problem}[Model-free Source Seeking]
Assume that the measurement signal in (\ref{eq:obs_model}) is scalar\footnote{The assumption is made only to simplify the presentation of the gradient ascent approach in the model-free case. The approach generalizes to signals of higher dimension.} and its expectation is maximized at the true state $y$ of the source:
\begin{equation}
\label{ass:scalar_measurement}
f^{-1}(y) = \argmax_{x \in \mathcal{X}} h(x,y).
\end{equation}
Generate a sequence of control inputs $u_0, u_1, \ldots$ for the team of sensors in order to drive its centroid $m_t$ towards a maximum of the signal field $h(\cdot,y)$.
\end{problem}


In the model-based case, the sensors have accurate knowledge of $h(\cdot,\cdot)$ which can be exploited to maximize the information that future measurements provide about the signal source. We choose \textit{mutual information} as a measure of informativeness. In order to select appropriate control inputs for the sensors at time $t-1$ we formulate the following optimization problem.
\begin{problem}[Model-based Source Seeking]
\label{prob:modelbased_ss}
Given the sensor poses $x_{t-1} \in \mathcal{X}^n$ and a prior distribution of the source state $y$ at time $t-1$, choose the control input $u_t \in \mathcal{U}^n$, which optimizes the following:
\begin{alignat}{2}
\max_{u_{1,t},\ldots,u_{n,t}} &I(y;z_t \mid x_t)&& \label{prob:max_MI}\\
\text{s.t.} \quad &x_{i,t} = x_{i,t-1}+u_{i,t}, \quad &&i = 1,\ldots, n, \notag\\
&z_{i,t} = h(x_{i,t},y) + v_{i,t}, \quad &&i = 1,\ldots, n. \notag
\end{alignat} 
\end{problem}

We resort to stochastic approximation methods in both scenarios and emphasize their usefulness in simplifying the algorithms while providing theoretic guarantees about the performance.





\section{Model-free Source Seeking}
\label{sec:model_free}
\subsection{Model-free Algorithm}


Our model-free approach is to design an iterative optimization scheme which causes the centroid $m_t$ of the robot formation to ascend the gradient $g(x,y) := \nabla_x h(x,y)$ of the measurement signal. The gradient ascent leads $m_t$ to a (often local) maximum of the signal field, which is appropriate in view of assumption (\ref{ass:scalar_measurement}). In detail, the desired dynamics for the centroid are:
\begin{equation}
\label{eq:mass_dyanmics}
m_{t+1} = m_t + \gamma_t g(m_t,y).
\end{equation}
A complication arises because the sensors do not have access to $g(\cdot,y)$ and can only measure a noisy version of $h(\cdot,y)$ at their current positions. Supposing noise-free measurements for now, the sensors can approximate the signal gradient at the formation centroid via a finite-difference (FD) scheme:
\begin{equation}
\label{eq:grad_approx}
g(m_t,y) = \nabla_x h(m_t,y) = W(x_t) \begin{pmatrix}
h(x_{1,t},y)\\
\vdots\\
h(x_{n,t},y)
\end{pmatrix} - b_t,
\end{equation}
where $W(x_t) \in \mathbb{R}^{d_x \times n}$ is a matrix of FD weights, which depends on the sensor states $x_t$, and $b_t \in \mathbb{R}^{d_x}$ captures the error in the approximation. The most natural way to obtain the FD weights is to require that the approximation is exact for a set of test functions $\psi_i, \; i = 1,\ldots,n$, commonly polynomials, which could represent the shape of $g(\cdot,y)$. In particular, the following relation needs to hold:
\begin{equation}
\label{eq:FD_weights}
\begin{bmatrix}
\psi_1(x_{1,t}) & \cdots & \psi_1(x_{n,t})\\
\vdots & & \vdots\\
\psi_n(x_{1,t}) & \cdots & \psi_n(x_{n,t})
\end{bmatrix} W(x_t)^T = \begin{bmatrix}
\frac{\partial}{\partial x} \psi_1(m_t)\\
\vdots\\
\frac{\partial}{\partial x} \psi_n(m_t)
\end{bmatrix},
\end{equation}
where $\frac{\partial}{\partial x} \psi_i(x)$ is a row vector of partial derivatives. When $x_{i,t} \in \mathbb{R}$ the most common set of test functions are the monomials $\psi_i(x) = x^{i-1}$, in which case (\ref{eq:FD_weights}) becomes a Vandermonde system. The standard (monomial) FD approach is problematic when the states $x_{i,t}$ are high-dimensional and not in a lattice configuration because the system in (\ref{eq:FD_weights}) becomes ill-conditioned. These difficulties are alleviated by using radial basis functions (RBFs) $\psi_i(x) := \phi(\|x - x_{i,t}\|)$ as test functions. In particular, using Gaussian RBFs, $\phi(d) := e^{-(\delta d)^2}$, with \textit{shape parameter} $\delta >0$, guarantees that (\ref{eq:FD_weights}) is non-singular \cite{Fornberg_CMA13}. Then, the FD weights obtained from (\ref{eq:FD_weights}) as a function of $x_t$ are:
\begin{equation}
\label{eq:rbffd_weights}
W(x_t) = R(x_t)^T\Phi(x_t)^{-T},
\end{equation}
where for $x \in \mathcal{X}^n$, we let $\Phi_{ij}(x) := e^{-\delta^2\|x_j-x_i\|_2^2}$ and
\begin{equation}
\label{eq:rbffd_rhs}
R(x) := \begin{bmatrix}
  2\delta^2 e^{-\delta^2\|x_1-\sum_{i=1}^nx_i/n\|_2^2}(x_1-\sum_{i=1}^nx_i/n)^T\\
  \vdots\\
  2\delta^2 e^{-\delta^2\|x_n-\sum_{i=1}^nx_i/n\|_2^2}(x_n-\sum_{i=1}^nx_i/n)^T
\end{bmatrix}.
\end{equation}

Since the measurements are noisy, sensor $i$ can observe only $z_{i,t}$ rather than $h(x_{i,t},y)$. As a result, the gradient ascent (\ref{eq:mass_dyanmics}) can be implemented only approximately via $g(m_t,y) \approx W(x_t) z_t$ instead of (\ref{eq:grad_approx}) and with the additional complication that the measurement noise makes the iterates $m_t$ random. Our stochastic model-free source seeking algorithm is:
\begin{align}
\label{eq:approx_mass_dyanmics}
m_{t+1} = m_t + \gamma_t W(x_t) z_t.
\end{align}

The convergence of similar source seeking schemes is often studied in a deterministic framework \cite{Ogren04_multirobot} by assuming that the noise can be neglected, which is difficult to justify. In the following section, we show that the center of mass $m_t$, following the dynamics (\ref{eq:approx_mass_dyanmics}) with appropriately chosen step-sizes $\gamma_t$, converges to a neighborhood of a local maximum of $h(\cdot,y)$. Assuming all-to-all communication or a centralized location, which receives all state and measurement information from the sensors, the stochastic gradient ascent (\ref{eq:approx_mass_dyanmics}) can be implemented as is. It requires that the sensors are localized relative to one another, i.e. in the inertial frame of one sensor, but not globally, in the world frame. Notably, it is also not important to maintain a rigid sensor formation as the correct FD weights necessary to combine the observations are re-computed at every measurement location. The next section shows that the only requirement is that the sensor team is not contained in a subspace of $\mathbb{R}^{d_x}$ when measuring (e.g. at least 3 non-collinear sensors are needed for $d_x = 2$).

\subsection{Convergence Analysis}
To carry out the convergence analysis of the stochastic gradient ascent in \eqref{eq:approx_mass_dyanmics}, we resort to the theory of stochastic approximations \cite{Kushner03_stochasticApproximation, Borkar_SA08}. It is sufficient to consider the following stochastic approximation (SA) algorithm:
\begin{equation}
\label{eq:sa_eq}
m_{t+1} = m_t + \gamma_t(g(m_t) + b_t + D_t),
\end{equation}
where $b_t$ is a bias term, $D_t$ is a random zero-mean perturbation, $\gamma_t$ is a small step-size, and $m_t$ is a random sequence whose asymptotic behavior is of interest. The main result is that the iterates $m_t$ in \eqref{eq:sa_eq} asymptotically follow the integral curves of the ordinary differential equation (ODE) $\dot{m} = g(m)$. Since in our case with a fixed source state $y$, $g(m) := \nabla_x h(m,y)$, the ODE method \cite{Ljung_TAC77}, \cite[Ch.2]{Borkar_SA08} shows that the iterates $\{m_t\}$ almost surely (a.s.) converge to the set $\{x \mid \nabla_x h(x,y) = 0\}$ of critical points of $h(\cdot,y)$ under the following assumptions\footnote{While assumptions (A1)-(A5) are sufficient to prove the convergence in our application, they are by no means the weakest possible. If necessary some can be relaxed using the results in stochastic approximation \cite{Kushner03_stochasticApproximation,Borkar_SA08}.}:
\begin{enumerate}
  \item[(A1)] The map $g$ is Lipschitz continuous\footnote{\label{ftn:lip_con}Given two metric spaces $(\mathcal{X},d_x)$ and $(\mathcal{G},d_g)$, a function $g:\mathcal{X} \to \mathcal{G}$ is \textit{Lipschitz continuous} if there exists a real constant $0\leq L < \infty$ such that: $d_g(g(x_1),g(x_2)) \leq Ld_x(x_1,x_2),\; \forall x_1,x_2\in\mathcal{X}$.}.
  \item[(A2)] Step-sizes $\{\gamma_t\}$ are positive scalars satisfying:\begin{center}$\sum_{t=0}^\infty \gamma_t = \infty$ and $\sum_{t=0}^\infty \gamma_t^2 < \infty$.\end{center}
  \item[(A3)] $\{D_t\}$ is martingale difference sequence with respect to the family of $\sigma$-algebras $\mathcal{F}_t := \sigma(m_0,D_s,0\leq s \leq t)$, i.e. $D_t$ is measurable with respect to $\mathcal{F}_t$, $\mathbb{E}[\|D_t\|] < \infty$, and $\mathbb{E}[D_{t+1} \mid \mathcal{F}_t] = 0$ almost surely (a.s.) for all $t \geq 0$. Also, $D_t$ is square-integrable with $\mathbb{E}[ \|D_{t+1}\|^2 \mid \mathcal{F}_t ] \leq K (1+\|m_t\|^2)$ a.s. for $t \geq 0$ and some constant $K > 0$.
  \item[(A4)] $\{m_t\}$ is bounded, i.e. $\sup_t \|m_t\| < \infty$ a.s.
  \item[(A5)] $\{b_t\}$ is bounded and $b_t \to 0$ a.s. as $t \to \infty$.
\end{enumerate}
The proposed source-seeking algorithm \eqref{eq:approx_mass_dyanmics} can be converted to the SA form \eqref{eq:sa_eq} as follows:
\begin{align*}
m_{t+1} &= m_t + \gamma_t W(x_t) z_t = m_t + \gamma_t W(x_t)\begin{pmatrix}
h(x_{1,t},y)+v_{1,t}\\
\vdots\\
h(x_{n,t},y)+v_{n,t}
\end{pmatrix}\\
&= m_t + \gamma_t\left(g(m_t,y) + b_t + W(x_t) v_t \right),
\end{align*}
where the second equality follows from \eqref{eq:grad_approx}. Assumption (A1) ensures that $\dot{m} = g(m,y)$ has a unique solution for any initial condition and any fixed source state $y$. Assumption (A2) can be satisfied by an appropriate choice of the step-size, e.g. $\gamma_t = 1/(t+1)$. The selection of proper step-sizes is an important practical issue that is not emphasized in this paper but is discussed at length in \cite{Atanasov_SourceSeeking_ICRA12,Kushner03_stochasticApproximation, Spall03_stochasticApproximation}. We can satisfy (A4) by requiring that the environment $\mathcal{X}$ of the sensors is bounded and if necessary use a projected version of the gradient ascent \cite[Ch.5.4]{Borkar_SA08}. This also ensures that the FD weights are bounded and in turn (A3) is satisfied:
\begin{align*}
\left(\mathbb{E}\|D_t\|_2\right)^2 \leq& \mathbb{E} \|D_t\|_2^2 = \mathbb{E}\left[ \|D_t\|_2^2 \mid \mathcal{F}_{t-1} \right] = \mathbb{E}\left[ \|W(x_t)v_t\|_2^2 \right]\\
\leq \|W&(x_t)\|_2^2 \mathbb{E}\|v_t\|_2^2 = \|W(x_t)\|_2^2 \sum_{i=1}^n \tr(\mathbb{E}[v_{i,t}v_{i,t}^T]) < \infty\\
\mathbb{E}\left[ D_t \mid \mathcal{F}_{t-1}\right]&=\mathbb{E}[W(x_t)v_t] = W(x_t) \mathbb{E} v_t = 0,
\end{align*}
since the measurement noise in \eqref{eq:obs_model} is uncorrelated in time and has zero mean and a finite second moment. Note that the error term in (\ref{eq:grad_approx}) violates (A5) because it does not converge to $0$. However, if we ensure that the sensor formation is not contained in a subspace of $\mathbb{R}^{d_x}$, then $b_t$ remains bounded by some $\epsilon_0 > 0$, i.e. $\sup_{t} \|b_t\| \leq \epsilon_0$. Then, the argument in \cite[Ch.5, Thm.6]{Borkar_SA08} shows that the iterates $m_t$ converge a.s. to a small neighborhood of a local maximum, whose size depends on $\epsilon_0$. The result is summarized below.

\begin{theorem}
Suppose that the gradient $g(x,y) = \nabla_x h(x,y)$ of the measurement signal is Lipschitz continuous\textsuperscript{\ref{ftn:lip_con}} in $x$, the step-sizes $\gamma_t$ in \eqref{eq:approx_mass_dyanmics} satisfy (A2), the sensor state space $\mathcal{X}$ is bounded, and the sensor formation is not contained in a subspace of $\mathbb{R}^{d_x}$ at the measurement locations. Then, algorithm \eqref{eq:approx_mass_dyanmics} converges to a small neighborhood around a local maximum of the signal field $h(\cdot,y)$.
\end{theorem}

\section{Model-based Source Seeking}
\label{sec:model_based}

\subsection{Model-based Algorithm}
\label{sec:cen_model_based}
In this section, we address Problem \ref{prob:modelbased_ss} assuming all-to-all communication. The sensors can follow the gradient of the cost function in (\ref{prob:max_MI}) to reach a local maximum:
\begin{equation}
\label{eq:grad_mi}
x_{t+1} = x_t + \gamma_t \nabla_x I(y;z_t|x) \vert_{x = x_t},
\end{equation}
where $\gamma_t$ is the step-size at time $t$. Let $p(z \mid y,x)$ denote the probability density function (pdf) of the measurement signal in \eqref{eq:obs_model}. Let $p_t(y)$ be the pdf used by the sensors at time $t$ to estimate the state of the source, which is assumed independent of $x_t$. The following theorem gives an expression for the mutual information (MI) gradient provided that $p(z \mid y,x)$ is differentiable with respect to the sensor configurations.
\begin{theorem}[\cite{Schwager_ISRR11} ] \label{thm:MIG}
Let random vectors $Y$ and $Z$ be jointly distributed with pdf $p(y,z \mid x)$, which is differentiable with respect to the parameter $x \in \mathcal{X}$. Suppose that the support of $p(y,z \mid x)$ does not depend on $x$. Then, the gradient with respect to $x$ of the mutual information between $Y$ and $Z$ is
\[
\nabla_x I(Y;Z|x) = \int \int \bigl(\nabla_x p(y,z \mid x)\bigr) \log \frac{p(z \mid y,x)}{p(z \mid x)} dy dz,
\]
where $p(z \mid y,x)$ and $p(z \mid x)$ are the marginal and the conditional pdfs of $Z$. 
\end{theorem}
Obtaining the MI gradient is computationally very demanding for two reasons. First, an approximate representation is needed for the continuous pdfs in the integral. Second, at time $t$ the integration is over the collection of all sensor measurements $z_t = \begin{bmatrix}z_{1,t}^T,\ldots,z_{n,t}^T\end{bmatrix}^T$, which can have a very high dimension in practice. As mentioned in Sec. \ref{sec:intro}, most existing work has focused on accurate approximations. However, Thm. \ref{thm:MIG} allows us to make a key observation:
\begin{align}
\nabla_x I(y;z_t|x_t) &\phantom{:}= \mathbb{E}[ \pi_t(z_t,x_t) \mid x_t] = \int_{\mathcal{Z}} \pi_t(z_t,x_t) p_t(z_z \mid x_t) dz_t,\notag\\
\text{where}\qquad&\label{eq:mi_expectation}\\
\pi_t(z,x) &:= \int_{\mathcal{Y}} \frac{\nabla_x p(z \mid y,x)}{p_t(z \mid x)} p_t(y) \log \frac{p(z \mid y,x)}{p_t(z \mid x)} dy,\notag\\
p_t(z \mid x) &:= \int_{\mathcal{Y}} p(z \mid y, x) p_t(y) dy,\notag
\end{align}
where the independence between $y$ and $x_t$ is used for the decomposition: $p_t(y,z \mid x) = p(z \mid y,x) p_t(y)$. Relying on the signal model, the sensors can simulate realizations of the random variable $z_t$, iid with pdf $p_t(z_t \mid x_t)$. Instead of computing the integral in \eqref{eq:mi_expectation} needed for the gradient ascent \eqref{eq:grad_mi}, we propose the following stochastic algorithm for model-based source seeking: 
\begin{equation}
\label{eq:sra}
x_{t+1} = x_t + \gamma_t \pi_t(z_{t},x_t).
\end{equation}
This algorithm can be written in the SA form \eqref{eq:sa_eq} as follows:
\begin{align*}
x_{t+1} &= x_t + \gamma_t \mathbb{E}_{z_{t}} [\pi_t(z_t,x_t) \mid x_t] + \gamma_t D_t\\
&= x_t + \gamma_t \left(\nabla_x I(y;z_t \mid x_t) + D_t\right),
\end{align*}
where $D_{t} := \pi_t(z_{t},x_t)-\mathbb{E} [\pi_t(z_{t},x_t) \mid x_t]$. To evaluate the convergence we consider assumptions (A1)-(A5) again. As before, satisfaction of (A2) is achieved by a proper step-size choice, while (A4) holds due to the bounded workspace $\mathcal{X}$. Assumption (A5) is satisfied because in this case the bias term is zero. To verify (A3), note that $D_t$ is measurable with respect to $\mathcal{F}_t = \sigma(x_0,D_s,0\leq s\leq t)$ and for $t \geq 1$:
\begin{align*}
\mathbb{E}[D_{t} \mid \mathcal{F}_{t-1}] &= \mathbb{E}\left[ \pi_t(z_{t},x_t)-\mathbb{E} [\pi_t(z_{t},x_t) \mid x_t] \mid \mathcal{F}_{t-1}\right]\\
&= \mathbb{E}[\pi_t(z_{t},x_t) \mid \mathcal{F}_{t-1}] - \mathbb{E} [\pi_t(z_{t},x_t) \mid x_t] = 0.
\end{align*}
Finally, if $p(z \mid y,x)$ and its gradient $\nabla_x p(z \mid y,x)$ are sufficiently regular (e.g. the former is bounded away from zero and the latter is Lipschitz continuous and bounded), the square integrability condition on $D_t$ is satisfied.

This analysis demonstrates that even if a single $z_t$ sample is used to approximate the MI gradient (instead of the integration in \eqref{eq:mi_expectation}), the stochastic gradient ascent \eqref{eq:sra} will converge to a local maximum of the mutual information between the source state and the sensor measurements.

\subsection{Implementation Details}

To implement the stochastic gradient ascent in (\ref{eq:sra}), the sensors need to propagate $p_t(\cdot)$ over time and sample from $p_t(\cdot \mid x_t)$. We achieve the first requirement by a particle filter \cite[Ch.4]{thrun_probrob05}, which approximates $p_t$ by a set of weighted samples $\{w_{t}^m,y_{t}^m\}_{m=1}^{N_p}$ as follows: $p_t(y) \approx \sum_{m=1}^{N_p} w_t^m \delta(y - y_t^m)$, where $\delta(\cdot)$ is a Dirac delta function. Using the particle set we can write $\pi$ and the measurement pdf as follows:
\begin{align*}
\pi_t(z,x) &\approx \sum_{m=1}^{N_p} w_t^m \frac{\nabla_x p(z \mid y_t^m,x)}{p(z\mid x)}  \log \frac{p(z \mid y_t^m,x)}{p(z \mid x)}\\
p_t(z \mid x) &\approx \sum_{m=1}^{N_p} w_t^m p(z \mid y_t^m,x),
\end{align*}
where $p(z \mid y,x)$ and its gradient can be decomposed further:
\begin{align*}
p(z \mid y, x) &= \prod_{j=1}^{n} p(z_j \mid y, x_j)\\
\frac{\partial p(z \mid y, x)}{\partial x_k} &=  \frac{\partial p(z_k \mid y, x_k)}{\partial x_k} \prod_{j\neq k} p(z_j \mid y, x_j) 
\end{align*}
due to the independence of the observations in \eqref{eq:obs_model}. In practice, there is a trade-off between moving the sensors and spending time approximating the gradient of the mutual information \eqref{eq:mi_expectation}. The stochastic approximation in \eqref{eq:sra} uses a single sample from $p_t(\cdot \mid x_t)$ but if sampling is fast compared to the time needed to relocate the sensors, more samples can be used to get a better estimate of the gradient. We use Monte Carlo integration, which proceeds as follows:
\begin{enumerate}
	\item Sample $\bar{m}(l)$ from the discrete distribution $w_t^1,\ldots, w_t^{N_p}$.
	\item Sample $\bar{z}(l)$ from the pdf $p(\cdot \mid y_t^{\bar{m}(l)},x_t)$.
	\item Repeat 1. and 2. to obtain $N_z$ samples $\{\bar{z}(l)\}_{l=1}^{N_z}$.
	\item Approximate: $\nabla_x I(y;z_t|x_t) \approx \frac{1}{N_z} \sum_{l=1}^{N_z} \pi_t(\bar{z}(l),x_t)$.
\end{enumerate}
Note that the advantage of improving the gradient estimate is not clear and should not necessarily be prioritized over the sensor motion. The SA techniques show that even an approximation with \textit{a single sample} is sufficient to make progress. In contrast, the related approaches mentioned in the introduction insist on improving the quality of the gradient estimate as much as possible. Depending on the application, this can slow down the robot motion and possibly make the algorithms impractical. Our more flexible approach adds an extra degree of freedom by allowing a trade-off between the gradient estimation quality and the motion speed of the sensors.

\section{Distributed Algorithms}
\label{sec:dist_alg}
In many scenarios, all-to-all communication is either infeasible or prone to failures. In this section, we present distributed versions of the model-free and the model-based algorithms. Since the model-free algorithm should be applicable to light-weight platforms with no global localization capabilities, the sensors use noisy relative measurements of their neighbors's locations to estimate the collective formation state. In the model-based case the sensors may spread around the environment and we are forced to assume that each agent is capable of estimating its own state $x_{i,t}$. We begin with preliminaries on distributed estimation.



\subsection{Preliminaries on Distributed Estimation}
\label{sec:dist_alg_prelim}
\label{sec:dist_estm}
Let the communication network of the $n$ sensors be represented by  an undirected graph $G = (\{1,\ldots,n\},E)$. Suppose that the sensors need to estimate an unknown static parameter $\theta^* \in \Theta$ in a distributed manner, where $\Theta \subseteq \mathbb{R}^{d_\theta}$ is a convex parameter space. At discrete times $k \in \mathbb{N}$, each agent $i$ observes a random signal $s_i(k) \in \mathbb{R}^{d_i}$ drawn from a distribution with conditional pdf $l_i(\cdot \mid \theta)$. Assume that the signals are iid over time and independent from the observations of all other sensors. The signals observed by a single agent, although potentially informative, do not reveal the parameter completely, i.e. each agent faces a local identification problem. We assume, however, that the parameter is identifiable if one has access to the signals observed by all agents. In order to aggregate the information provided to it over time - either through observations or communication with neighbors - each sensor $i$ holds and updates a pdf $p_{i,k}:\Theta \rightarrow \mathbb{R}_{\geq 0}$ over the parameter space as follows:
\begin{align}
p_{i,k+1}(\theta) &= \eta_{i,k} l_i(s_i(k+1) \mid \theta) \prod_{j \in \mathsf{N}_i \cup \{i\}} \bigl(p_{j,k}(\theta) \bigr)^{a_{ij}} \notag\\
\hat{\theta}_{i}(k) &\in \argmax_{\theta \in \Theta} p_{i,k}(\theta), \label{eq:dist_estm}
\end{align}
where $\eta_{i,k}$ is a normalization constant ensuring that $p_{i,k+1}$ is a proper pdf, $\mathsf{N}_i$ is the set of nodes (neighbors) connected to sensor $i$, and $a_{ij}$ are weights such that $\sum_{j \in \mathsf{N}_i \cup \{i\}} a_{ij} = 1$. The update is the same as the standard Bayes rule with the exception that sensor $i$ does not just use its own prior but a \textit{geometric average} of its neighbors' priors. Given that $G$ is connected, the authors of \cite{Alireza_cdc2010} show that the distributed estimator (\ref{eq:dist_estm}) is weakly consistent\footnote{\label{ftn:weak_consistency}Weak consistency means that the estimates $\hat{\theta}_{i}(k)$ converge in probability to $\theta^*$, i.e. $\displaystyle{\lim_{k\to\infty} \mathbb{P}\bigl(\|\hat{\theta}_{i}(k) - \theta^*\| \geq \epsilon \bigr) = 0}$ for any $\epsilon > 0$ and all $i$.} under broad assumptions on the signal models $l_i(\cdot \mid \theta)$. The results in \cite{shahrampour_cdc13, Alireza_TAC10} suggest that this algorithm is even applicable to a time-varying graph topology with asynchronous communication.



\paragraph{Specialization to Gaussian distributions}
We now specialize the general scheme of Rad and Tahbaz-Salehi \cite{Alireza_cdc2010} to Gaussian distributions. To our knowledge, this specialization is new and the theorem obtained below (Thm. \ref{thm:dist_estm_L2}) shows that the resulting distributed linear Gaussian estimator is mean-square consistent\footnote{Mean-square consistency means that the estimates $\hat{\theta}_i(k)$ converge in mean-square to $\theta^*$, i.e. $\displaystyle{\lim_{k\to\infty}\mathbb{E} \left[\|\hat{\theta}_i(k) - \theta^*\|^2 \right]=0}$ for all $i$.}, which is stronger than the weak consistency\textsuperscript{\ref{ftn:weak_consistency}} shown in \cite[Thm.1]{Alireza_cdc2010}. Suppose that the agents' measurement signals are \textit{linear} in the parameter $\theta^*$ and perturbed by Gaussian noise:
\begin{equation}
\label{eq:lin_ga_obs}
s_{i}(k) = H_i \theta^* + \epsilon_{i}(k), \quad \epsilon_{i}(k) \sim \mathcal{N}(0,E_i), \quad \forall i.
\end{equation}
Let $\mathcal{G}(\omega,\Omega)$ denote a Gaussian distribution (in information space) with mean $\Omega^{-1}\omega$ and covariance matrix $\Omega^{-1}$. Since the private observations (\ref{eq:lin_ga_obs}) are linear Gaussian, without loss of generality the pdf $p_{i,k}$ of agent $i$ is the pdf of a Gaussian $\mathcal{G}(\omega_{i,k},\Omega_{i,k})$. Exploiting that the parameter $\theta^*$ is static, the update equation of the distributed filter in (\ref{eq:dist_estm}), specialized to Gaussian distributions, is:
\begin{align}
\omega_{i,k+1} &\phantom{:}= \sum_{j \in \mathsf{N}_i \cup \{i\}} a_{ij} \omega_{j,k} + H_i^T E_i^{-1} s_{i}(k),\notag\\
\Omega_{i,k+1} &\phantom{:}= \sum_{j \in \mathsf{N}_i \cup \{i\}} a_{ij} \Omega_{j,k} + H_i^T E_i^{-1} H_i,\notag\\
\hat{\theta}_i(k) &:= \Omega_{i,k}^{-1}\omega_{i,k}.\label{eq:gauss_dist_estm}
\end{align}
In this linear Gaussian case, we prove (Appendix B) a strong result about the quality of the estimates in \eqref{eq:gauss_dist_estm}.
\begin{theorem}
\label{thm:dist_estm_L2}
Suppose that the communication graph $G$ is connected and the matrix $\begin{bmatrix} H_1^T & \hdots & H_n^T \end{bmatrix}^T$ has rank $d_\theta$. Then, the estimates \eqref{eq:gauss_dist_estm} of all agents converge in mean square to $\theta^*$, i.e. $\displaystyle{\lim_{k \to \infty} \mathbb{E} \left[ \|\hat{\theta}_{i}(k) - \theta^*\|^2\right] = 0, \; \forall i}$.
\end{theorem}



\paragraph{Specialization to particle distributions} 
Suppose that the pdf $p_{i,k}$ is represented by a set of particles $\{w_{i,k}^m, \theta_{i,k}^m\}_{m=1}^{N_p}$, which are \textit{identical} for all sensors initially (at $k = 0$). Since the parameter $\theta^*$ is stationary, the particle positions $\theta_{i,k}^m$ will remain the same across the sensors for all time. The update equation of the distributed filter in \eqref{eq:dist_estm}, specialized to particle distributions, only needs to propagate the particle importance weights $w_{i,k}^m$ and is summarized in Alg. \ref{alg:dist_part_filt}.
\begin{algorithm}[htb]
\caption{Distributed Particle Filter at Sensor $i$}
\begin{algorithmic}[1]
\footnotesize
\State \textbf{Input}: Particle sets $\{w_{j,k}^m, \theta_{j,0}^m\}$ for $m=1,\ldots,N_p$ and $j \in \mathsf{N}_i \cup \{i\}$, private signal $s_{i}(k+1)$, and pdf $l_i(\cdot\mid \cdot)$
\State \textbf{Output}: Particle weights $\{w_{i,k+1}^m\}$ for $m=1,\ldots,N_p$
\State Average priors: $\bar{w}_{i,k}^m \gets \exp \left( \sum_{j \in \mathsf{N}_i \cup \{i\}} a_{ij} \log(w_{j,k}^m) \right)$
\State Update: $w_{i,k+1}^m \gets \bar{w}_{i,k}^m l_i(s_i(k+1) \mid \theta_{i,0}^m)$ for $m=1,\ldots,N_p$
\State Normalize the weights
\State \textbf{return} $\{w_{i,k+1}^m\}$ for $m=1,\ldots,N_p$
\end{algorithmic}
\label{alg:dist_part_filt}
\end{algorithm}

\subsection{Distributed Model-free Algorithm}
\label{sec:dist_mf}
To distribute the model-free algorithm \eqref{eq:approx_mass_dyanmics}, the sensor formation needs to estimate its configuration $x_t$, the centroid $m_t$, and the stochastic approximation to the signal gradient $W(x_t)z_t$ at each measurement location (i.e. at each time $t$) using only local information. We introduce a fast time-scale $k = 0,1,\ldots$, which will be used for the estimation procedure at each time $t$. During this the sensors remain stationary and we drop the $t$ index to simplify the notation. As mentioned earlier, we suppose that each sensor $i$ receives a relative measurement of the state of each of its neighbors $j \in \mathsf{N}_i$:
\begin{equation}
\label{eq:rel_mes}
s_{ij}(k) = x_j - x_i + \epsilon_{ij}(k), \quad \epsilon_{ij}(k) \sim \mathcal{N}(0,E_{ij}),
\end{equation}
where $\epsilon_{ij}(k)$ is the measurement noise which is independent at any pair of times on the fast time-scale and across sensor pairs. If each sensor manages to estimate the states of the whole sensor formation using the measurements $\{s_{ij}(k)\}$, then each can compute the FD weights in (\ref{eq:rbffd_weights}) on its own.

The distributed linear Gaussian estimator \eqref{eq:gauss_dist_estm} can be employed to estimate the sensor states $x$. Notice that it is sufficient to estimate $x$ in a local frame because neither the finite difference computation (\ref{eq:rbffd_weights}) nor the gradient ascent (\ref{eq:approx_mass_dyanmics}) requires global state information. Assume that all sensors know that sensor 1 is the origin at every measurement location. Let $x^* := \begin{bmatrix}0^T & (x_2-x_1)^T & \cdots & (x_n-x_1)^T \end{bmatrix}^T$ denote the true sensor states in the frame of sensor $1$. Let $\hat{x}^i(k)$ denote the estimate that sensor $i$ has of $x^*$ at time $k$ on the fast time scale. The vector form of the measurement equations \eqref{eq:rel_mes} is:
\begin{equation}
\label{eq:rel_mes_vec}
s(k) = (B \otimes I_{d_x})^T x^* + \epsilon(k),
\end{equation}
where $B$ is the incidence matrix of the communication graph $G$. The measurements \eqref{eq:rel_mes_vec} fit the linear Gaussian model in (\ref{eq:lin_ga_obs}). Since the first element of $x^*$ is always 0, only $(n-1)d_x$ components need to be estimated. As the rank of $B \otimes I_{d_x}$ is also $(n-1)d_x$, Thm. \ref{thm:dist_estm_L2} allows us to use the distributed estimator \eqref{eq:gauss_dist_estm} to update $\hat{x}^i(k)$.

Concurrently with the state estimation, sensor $i$ would be obtaining observations $z_{i,t}(k)$ of the signal field for $k = 0,1,\ldots$\footnote{The time-scales of the relative state measurements and the signal measurements might be different but for simplicity we keep them the same.}. In the centralized case (Sec. \ref{sec:model_free}), each sensor uses the following gradient approximation:
\begin{equation}
\label{eq:grad_decompose}
g(m_t,y) \approx W(x_t) z_t = \sum_{i=1}^n \mathbf{col}_i(W(x_t)) z_{i,t},
\end{equation}
where $\mathbf{col}_i(W(x_t))$ denotes the $i$th column of the FD-weight matrix. Since $x_t$ and $z_t$ are not available in the distributed setting, each sensor can use its local measurements $z_{i,t}(k)$ and its estimate $\hat{x}_t^i(k)$ of the sensor states to form its own local estimate of the signal gradient:
\begin{equation}
\label{eq:local_grad_estm}
\hat{g}_{i,t}(k) := \mathbf{col}_i(W(\hat{x}_t^i(k))) \frac{1}{k+1} \sum_{\tau=0}^k z_{i,t}(\tau).
\end{equation}
In order to obtain an approximation to $g(m_t,y)$ as in \eqref{eq:grad_decompose} in a distributed manner, we use a high-pass dynamic consensus filter \cite{Spanos_IFAC05} to have the sensors agree on the value of the sum:
\[
\hat{g}_t(k) := n \biggl(\frac{1}{n}\sum_{i=1}^n \hat{g}_{i,t}(k)\biggr).
\]
Each node maintains a state $q_{i,k}$, receives an input $\mu_{ik}$, and provides an output $r_{ik}$ with the following dynamics:
\begin{align}
q_{i,k+1} &= q_{i,k} + \beta \sum_{j \in \mathsf{N}_i} (q_{j,k} - q_{i,k}) + \beta \sum_{j \in \mathsf{N}_i} (\mu_{j,k} - \mu_{i,k}) \notag\\
r_{i,k} &= q_{i,k} + \mu_{i,k} \label{eq:consensus_filter}
\end{align}
where $\beta >0$ is a step-size. For a connected network \cite[Thm.1]{Spanos_IFAC05} guarantees that $r_{i,k}$ converges to $1/n \sum_i \mu_{i,k}$ as $k \to \infty$. The following result can be shown by letting $\mu_{i,k} := \hat{g}_{i,t}(k)$ and is proved in the appendix.
\begin{theorem}
\label{thm:unbiased_grad_estimate}
Suppose that the communication graph $G$ is strongly connected. If the sensor nodes estimate their states $x^*$ from the relative measurements (\ref{eq:rel_mes_vec}) using algorithm \eqref{eq:gauss_dist_estm}, compute the FD-weights (\ref{eq:rbffd_weights}) using the state estimates, and run the dynamic consensus filter (\ref{eq:consensus_filter}) with input $\mu_{i,k} := \hat{g}_{i,t}(k)$, which was defined in (\ref{eq:local_grad_estm}), then the output $r_{i,k}$ of the consensus filter satisfies:
\[
n \left(\lim_{k\to \infty} \mathbb{E}[r_{i,k}]\right) = g(m^*,y) + b, \quad \forall i \in \{1,\ldots,n\},
\]
where $g(m^*,y)$ is the true signal gradient at $m^* := \sum_{i=1}^n x_i^*/n$ and $b$ is the error in the finite-difference approximation (\ref{eq:grad_approx}).
\end{theorem}

After this procedure the agents agree on a centroid for the formation and a gradient estimate, which can be used to compute the next formation centroid according to (\ref{eq:approx_mass_dyanmics}). Since the FD weights are re-computed at every $t$, the formation need not be maintained accurately. This allows the sensors to avoid obstacles and takes care of the motion uncertainty.

\subsection{Distributed Model-based Algorithm}
\label{sec:dist_mb}
In this section, we aim to distribute the model-based source-seeking algorithm (\ref{eq:sra}). We make an assumption that sensors which are far from each other receive independent information. This is reasonable for sensors with limited sensing range because when they are far from each other, their sensed signals (if any) would not be coming from the same source. As a result, computing the  mutual information gradient in \eqref{eq:mi_expectation} with respect to $x_i$ is decoupled from the states of the distant sensors.
\begin{theorem}
\label{thm:MI_decouple}
Let $\mathcal{V}_i$ denote the set of sensors (excluding $i$) whose fields of view overlap with that of sensor $i$. Let $\bar{\mathcal{V}}_i$ denote the rest of the sensors. Suppose that sensor $i$'s measurements, $z_i$, are independent (not conditionally on $y$, as before) of the measurements, $z_{\bar{V}_i}$, obtained by the sensors $\bar{V}_i$, i.e. $p_t(z_i,z_{\bar{V}_i} \mid x_i, x_{\bar{V}_i}) = p_t(z_i \mid x_i) p_t(z_{\bar{V}_i} \mid x_{\bar{V}_i})$. Then:
\begin{center}$
\frac{\partial}{\partial x_i} I\bigl(y; z_i, z_{\mathcal{V}_i}, z_{\bar{\mathcal{V}}_i} \mid x_i, x_{\mathcal{V}_i}, x_{\bar{\mathcal{V}}_i}\bigr) = \frac{\partial}{\partial x_i} I\bigl(y; z_i, z_{\mathcal{V}_i} \mid x_i, x_{\mathcal{V}_i}\bigr).
$\end{center}
\end{theorem}

\begin{proof}
By the chain rule of mutual information and then the independence of $z_i$ and $z_{\bar{\mathcal{V}}_i}$:
\begin{align*}
I\bigl(y; z_i, z_{\mathcal{V}_i}&, z_{\bar{\mathcal{V}}_i}\mid x_i, x_{\mathcal{V}_i}, x_{\bar{\mathcal{V}}_i}\bigr)\\
& = I\bigl(y; z_i, z_{\mathcal{V}_i} \mid x_i, x_{\mathcal{V}_i}\bigr) + I\bigl(y; z_{\bar{\mathcal{V}}_i} \mid z_i, z_{\mathcal{V}_i}, x_i, x_{\mathcal{V}_i}, x_{\bar{\mathcal{V}}_i}\bigr)\\
& = I\bigl(y; z_i, z_{\mathcal{V}_i} \mid x_i, x_{\mathcal{V}_i}\bigr) + I\bigl(y; z_{\bar{\mathcal{V}}_i} \mid z_{\mathcal{V}_i}, x_{\mathcal{V}_i}, x_{\bar{\mathcal{V}}_i}\bigr).
\end{align*}
The second term above is constant with respect to $x_i$. \hfill $\blacksquare$
\end{proof}

As a result of Thm. \ref{thm:MI_decouple} and the stochastic approximation algorithm in (\ref{eq:sra}), sensor $i$ updates its pose as follows:
\begin{equation}
\label{eq:dist_MIgrad_ascend}
x_{i,t+1} = x_{i,t} + \gamma_t \pi_t\bigl(z_{\{i\} \cup \mathcal{V}_i,t}, x_{\{i\} \cup \mathcal{V}_i,t}\bigr).
\end{equation}
This update is still not completely distributed as it requires knowledge of $x_{\mathcal{V}_i,t}$ and the pdf $p_t$\footnote{Since all sensors have the same observation model $h(\cdot,\cdot)$, each sensor can simulate measurements $z_{\mathcal{V}_i,t}$ as long as it knows the configurations $x_{\mathcal{V}_i,t}$.}. We propose to distribute the computation of $p_t$ via the distributed particle filter (Alg. \ref{alg:dist_part_filt}). Then, each sensor maintains its own estimate of the source pdf, $p_{i,t}$, represented by a particle set $\{w_{i,t}^m,y_{i,t}^m\}$. Given a new measurement, $z_{i,t+1}$, sensor $i$ averages its prior, $p_{i,t}$, with the priors of its neighbors and updates it using Bayes rule. Finally, to obtain $x_{\mathcal{V}_i,t}$ we use a flooding algorithm (Alg. \ref{alg:conf_exhange}). The convergence analysis of the gradient ascent scheme in the distributed case \eqref{eq:dist_MIgrad_ascend} remains the same as in Sec. \ref{sec:model_based} because each sensor $i$ computes the complete MI gradient. This is possible because due to Thm. \ref{thm:MI_decouple} the states and measurements of distant sensors are not needed, while Alg. \ref{alg:conf_exhange} provides the information from the nearby sensors.

\begin{algorithm}[htb]
\caption{States Exchange Algorithm at Sensor $i$}
\begin{algorithmic}[1]
\footnotesize
\State \textbf{Input}: Communication radius $r_c$, sensing radius $r_s$, state $x_i$
\State \textbf{Output}: Array $a_i$ with $a_i[j] = x_j$ if $j \in \mathcal{V}_i\cup\{i\}$ and $a_i[j] = \textit{empty}$ else
\State $a_i[i] \gets x_i, \quad a_i[j] \gets \textit{empty}, \;\; j \neq i$ \Comment Holds the required sensor states
\State $b \gets \min\{ \mathbf{ceil}(2 r_s/r_c), n\}$ \Comment Number of rounds needed
\For{$k = 1 \ldots b$}
  \State Send $a_i$ to neighbors $\mathsf{N}_i$, receive $\{a_j\}$ from $j \in \mathsf{N}_i$ 
  \For{$j \in \mathsf{N}_i$}
    \For{$l = 1 \ldots n$}
      \If{$(a_i[l] = \textit{empty}) \&\& (a_j[l] \neq \textit{empty})$}
        \State $a_i[l] \gets a_j[l]$
      \EndIf
    \EndFor
  \EndFor
\EndFor
\end{algorithmic}
\label{alg:conf_exhange}
\end{algorithm}

\section{Applications}
\label{sec:applications}
\begin{figure*}[t!]
	\begin{center}
		\includegraphics[width=\linewidth]{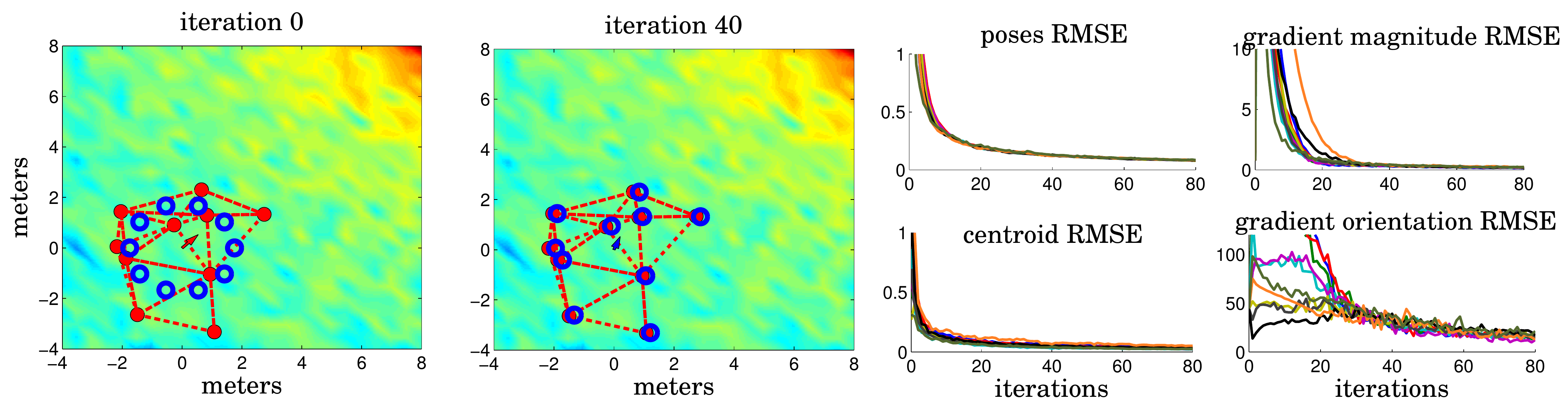}
	\end{center}
	\caption{Joint position and gradient estimation at a single measurement location (on the fast time-scale). The first plot shows the true sensor positions (red circles), initial position estimates (blue circles), and the true gradient of the signal field (red arrow). The second plot shows the position estimates after 40 iterations (blue circles) and the gradient estimate of sensor 1 (blue arrow). The third column shows the root mean squared error (RMSE) of the position (top) and centroid (bottom) estimates of all sensors averaged over 50 independent repetitions. The forth column shows the RMSE of the gradient magnitude and orientation estimates.}
	\label{fig:pose_grad_estm}
\end{figure*}

The performance of the source-seeking algorithms is demonstrated in simulation using a team of $10$ sensors to localize the source of a wireless radio signal. We consider a radio signal because it is very noisy and difficult to model and yet most approaches for wireless source seeking are model-based, which makes it suitable for comparing the two algorithms. We begin by modeling the received signal strength (RSS), which is needed for the model-based algorithm.

\subsection{RSS Model}
\label{sec:rss_model}
Let the positions of a wireless source and receiver in 2D be $y$ and $x$, respectively. The \textit{received signal strength} (dBm) at $x$ is modeled as:
\begin{center}$\begin{aligned}
P_{rx}(x,y) = P_{tx} &+ G_{tx} - L_{tx}  + G_{rx} - L_{rx} \\
&- L_{fs}(x,y) - L_{m}(x,y) - R(x,y),
\end{aligned}$\end{center}
where $P_{tx}$ is the transmitter output power (18 dBm in our experiments), $G_{tx}$ is the transmitter antenna gain (1.5 dBi), $L_{tx}$ is the transmitter loss (0 dB), $G_{rx}$ is the receiver antenna gain (1.5 dBi), $L_{rx}$ is the receiver loss (0 dB), $L_{fs}$ is the free space loss (dB), $L_m$ is the multi-path loss (dB), and $R$ is the noise. The free space loss is modeled as:
\begin{center}$
L_{fs}(x,y) = -27.55 + 20 \log_{10}(\nu) + 20 \log_{10}\bigl(\|x-y\|_2\bigr),
$\end{center}
where $\nu$ is the frequency (2400 MHz). The model from \cite{Capulli06_wirelessNetworks} is used for the the multi-path loss:
\begin{center}$
L_{m}(x,y) = \begin{cases}
             \alpha + \beta \lambda(x,y),&\text{if } \lambda(x,y) > 0\\
             0,& \text{else}
             \end{cases}
$\end{center}
where $\alpha$ is a multi-wall constant (30 dB), $\beta$ is a wall attenuation factor (15 dB/m), and $\lambda(x,y)$ denotes the distance traveled by the ray from $y$ to $x$ through occupied cells in the environment (represented as an occupancy grid). Finally, if the measurement is line-of-sight (LOS), i.e. $\lambda(x,y) = 0$, the fading $R(x,y)$ is Rician$(\mu,\sigma)$; otherwise it is Rayleigh$(\sigma)$. We used $\mu = 4$ dB and $\sigma = 20$ dB in the simulations.

\subsection{Simulation Results}
The first experiment aims at verifying the conclusions of Thm. \ref{thm:unbiased_grad_estimate} when the sensor formation is not maintained well, namely that the distributed relative pose estimation and the consensus on the local finite difference gradient estimates converge asymptotically to an unbiased (up to the error in the FD approximation) gradient estimate. Ten sensors were arranged in a distorted ``circular'' formation (see Fig. \ref{fig:pose_grad_estm}) and were held stationary during the estimation procedure (on the fast time-scale). Initially, the sensors assumed that they were in a perfect circular formation of radius 1.75 meters. Relative measurements (\ref{eq:rel_mes}) with noise covariance $E_{ij} = 0.4 I_2$ were exchanged to estimate the sensor states. At each time $k$, sensor $i$ used its estimate $\hat{x}^i(k)$ to compute the FD weights via (\ref{eq:rbffd_weights}). Wireless signal measurements obtained according to the RSS model were combined with the FD weights to form the local gradient estimates (\ref{eq:local_grad_estm}), which were used to update the state of the consensus filter according to (\ref{eq:consensus_filter}). Fig. \ref{fig:pose_grad_estm} shows that the errors in the pose and the gradient estimates tend to zero after $80$ iterations on the fast time scale.

\begin{figure}[t!]
  \begin{center}
      \includegraphics[width=\linewidth]{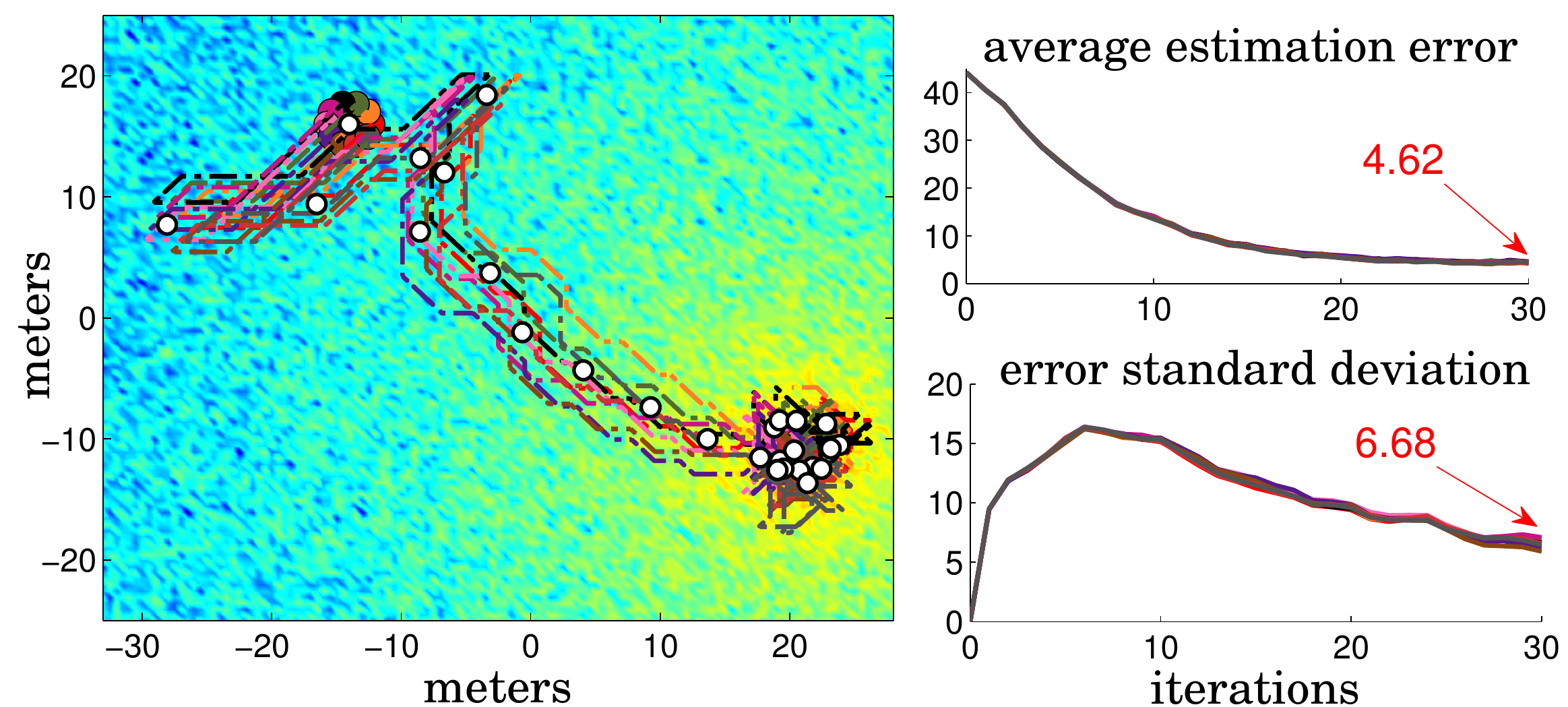}
  \end{center}
	\caption{The paths followed by the sensors after 30 iterations of the model-free source-seeking algorithm in an obstacle-free environment. The white circles indicate sensor 1's estimates of the source position over time. The plots on the right show the average error of the source position estimates and its standard deviation averaged over 50 independent repetitions.}
	\label{fig:mf}
\end{figure}
\begin{figure*}[t!]
	\begin{center}
	  \includegraphics[width=0.49\linewidth]{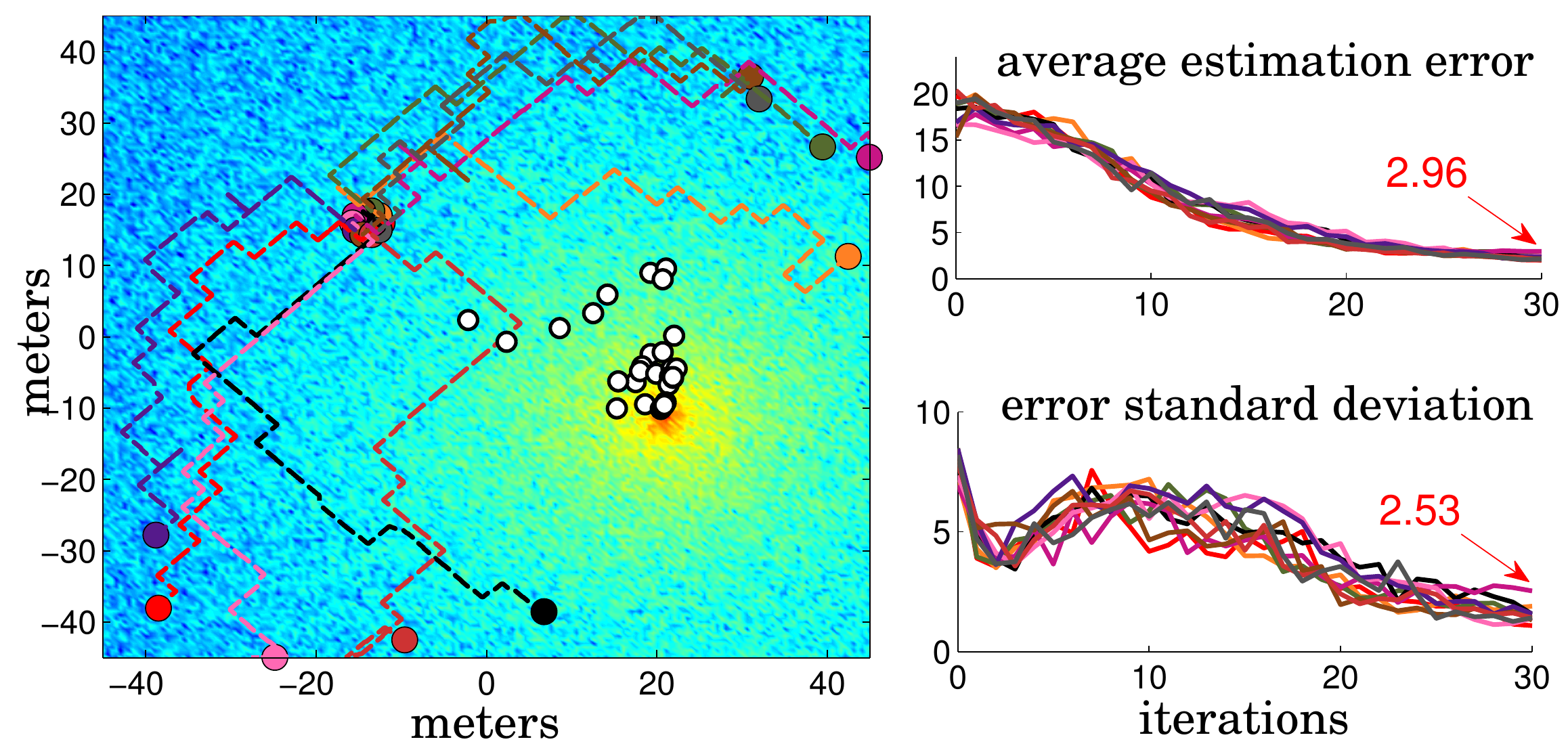}
		\includegraphics[width=0.49\linewidth]{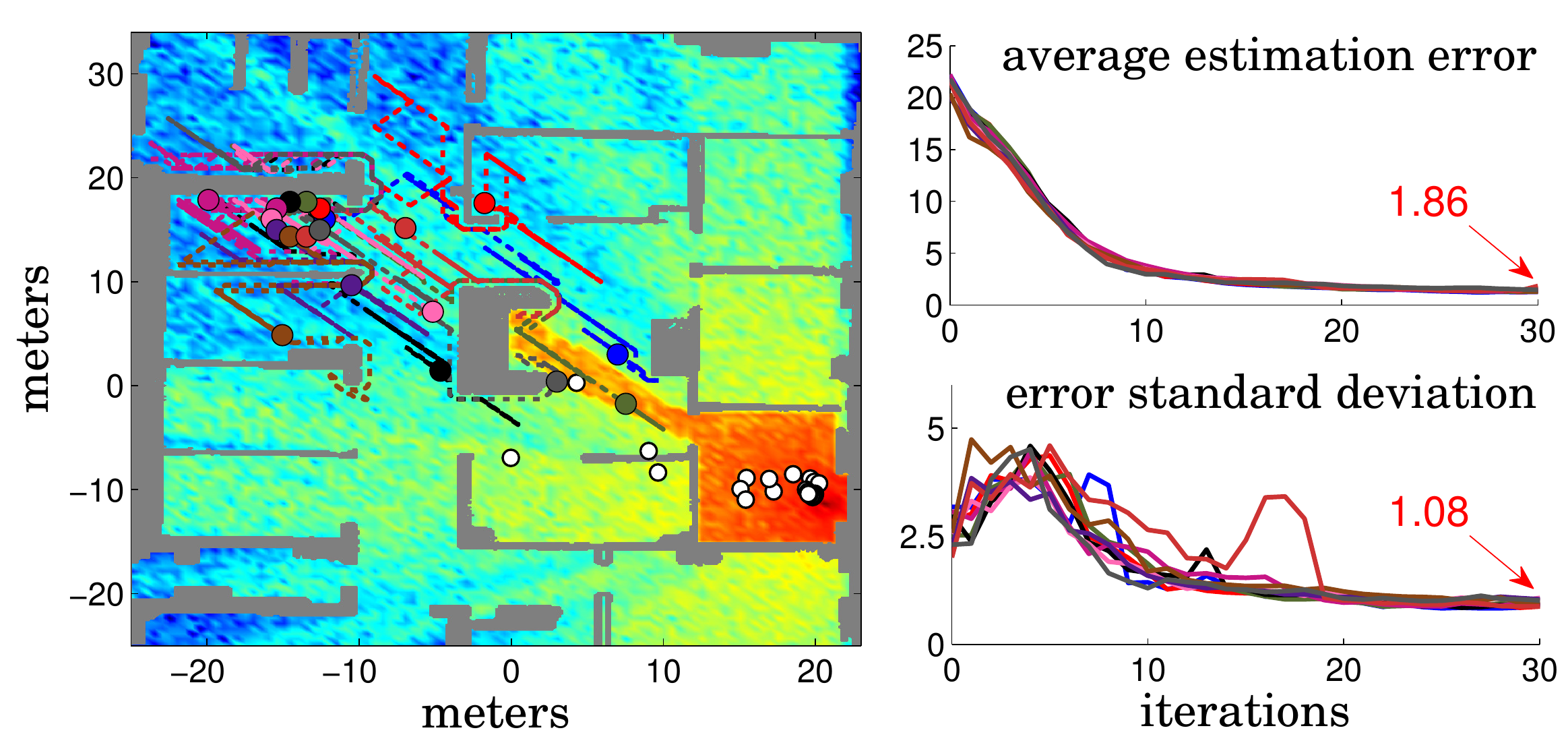}\\
		\includegraphics[width=0.495\linewidth]{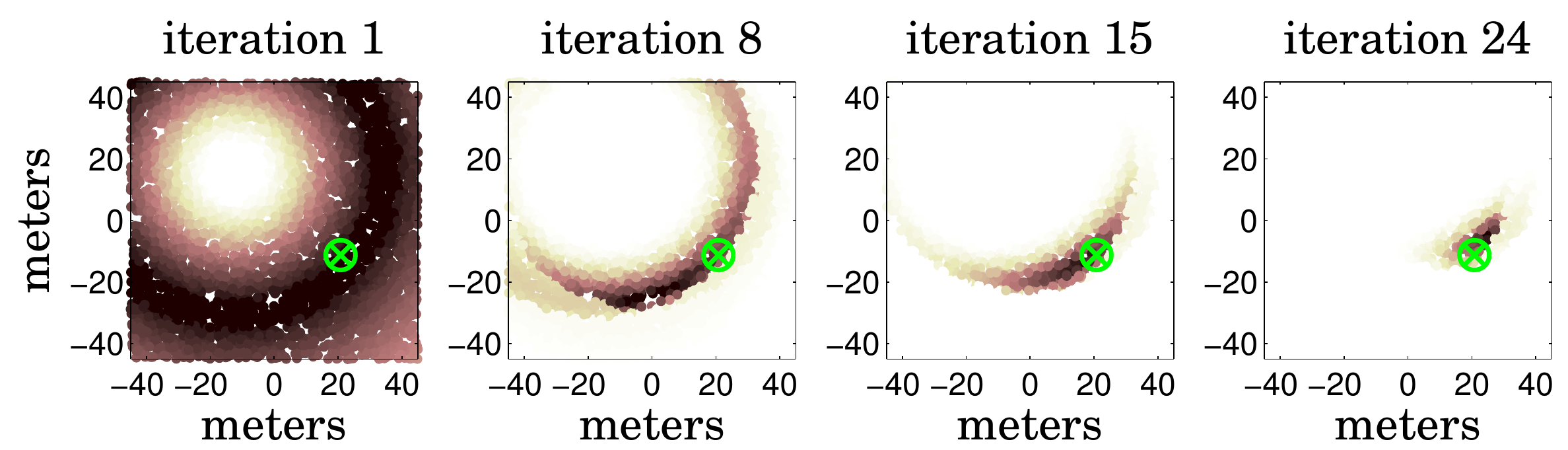}
		\includegraphics[width=0.495\linewidth]{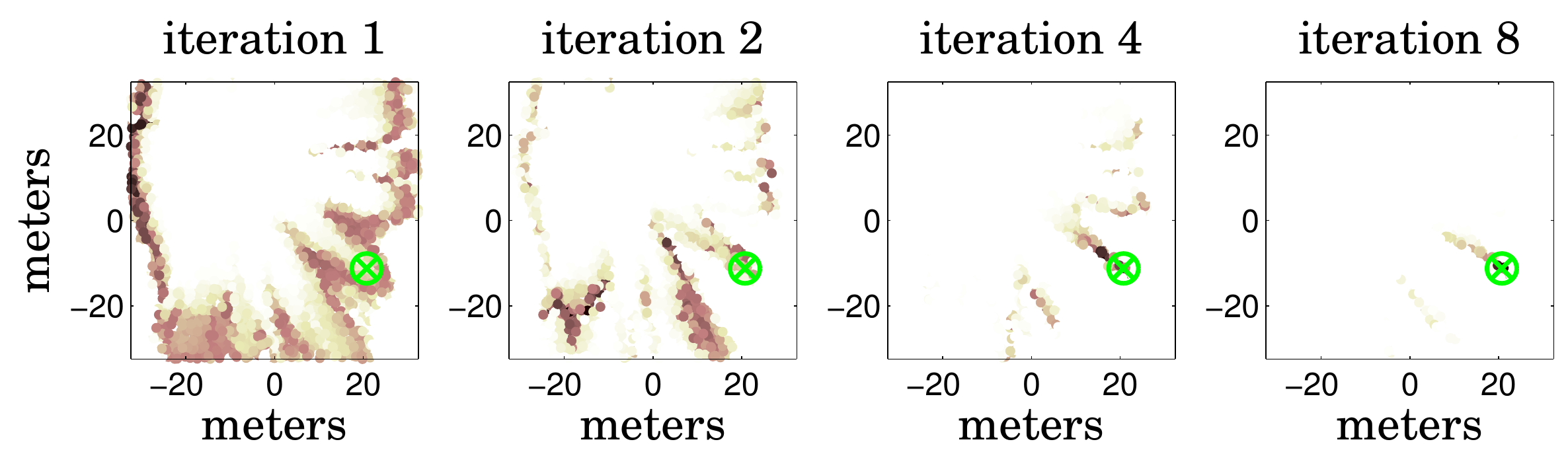}
	\end{center}
	\caption{The paths followed by the sensors after 30 iterations of the model-based source-seeking algorithm in an environment without obstacles (left) and with obstacles (right). The white circles indicate sensor 1's estimates of the source position over time. The plots show the average error of the source position estimates and its standard deviation averaged over 50 independent repetitions. The evolution of sensor 1's distributed particle filter is shown in each scenario (bottom row).}
	\label{fig:mb}
\end{figure*}

Next, we demonstrate the ability of our algorithms to localize the source of a wireless signal obtained using the RSS model of Sec. \ref{sec:rss_model}. The performance of the model-free algorithm is illustrated in Fig. \ref{fig:mf}. A circular formation with radius 1.75 meters consisting of 10 sensors was maintained. The communication radius was 6 meters, while the sensing radius was infinite. The sensors did not coordinate to maintain the formation. They were kept together by the agreement on the centroid and the signal gradient, achieved via the distributed state estimation and the consensus filter. At time $t$, each sensor $i$ applied the control $u_{i,t} = \gamma_t \hat{g}_{i,t}(K_{max})$, where $\hat{g}_{i,t}(K_{max})$ is the gradient estimate after $K_{max} = 50$ iterations on the fast time-scale and $\gamma_t$ is the step-size. Unlike the continuous measurements illustrated in Fig. \ref{fig:pose_grad_estm}, the sensors measured their relative states and the wireless signal only 10 times and stopped updating their local gradient estimates so that the consensus filter could converge fast. The initial distance between the signal source and the centroid of the sensor formation was 44.2 meters. Averaged over 50 independent repetitions the sensors managed to estimate the source location within 4.62 meters in 30 iterations. 

The same initial source and sensor positions were used to set up the model-based experiments. Fig. \ref{fig:mb} illustrates the performance in environments with and without obstacles. The communication radius was 10 meters, while the sensing radius was infinite. The sensors maintained distributed particle filters with 4000 particles and used 5 signal measurements to update the filters before moving (unlike the 10 used in the model-free case). The stochastic mutual information gradient was obtained via 10 simulated signal measurements only. Averaged over 50 independent repetitions, the sensors managed to estimate the source location within 2.96 meters in the obstacle-free case and 1.86 meters in the obstacle case after 30 iterations of the algorithm. It is interesting to note that the performance of the model-based algorithm is better when obstacles are present than in the obstacle-free environment. When the model is good and the environment is known, the wall attenuation of the signal helps the sensors discount many hypothetical source positions, which would not be possible in the obstacle-free case (see the distributed particle filter evolution in Fig. \ref{fig:mb}). We note that when a good signal model is available, the model-based algorithm outperforms the model-free one. However, we expect that as the quality of the model degrades so would the performance of the model-based approach and the model-free algorithm would become more attractive.

\section{Conclusion}
This paper presented distributed model-free and model-based approaches for source seeking with a mobile sensor network. Our stochastic gradient ascent approach to model-free source seeking does not necessitate global localization and is robust to deformations in the geometry of the sensor team. The stochastic approximation simplifies the algorithm and provides convergence guarantees. The model-based method has the sensors follow a stochastic gradient of the mutual information between their expected measurements and their source estimates. In this case, the stochastic approximation enables a key trade-off between time spent moving the sensors and time spent planning the most informative move. The experiments show that the model-based algorithm outperforms the model-free one when an accurate model of the signal is available. Its drawbacks are that it relies on knowledge of the environment, global localization, and a flooding algorithm to exchange the sensor states, which can be demanding for the network. If computation is limited, the environment is unknown, the signal is difficult to model, or global localization is not available, the model-free algorithm would be the natural choice. Future work will focus on comparing the performance of the algorithms with other source seeking algorithms in the literature. It is of interest to apply the algorithms to other signals and to carry out real-world experiments as well.


\appendix       
\section*{Appendix A: Proof of Thm. \ref{thm:unbiased_grad_estimate}}
From Thm. \ref{thm:dist_estm_L2}, $\hat{x}^{i}(k) \overset{L^2}{\rightarrow} x^*, \; \forall i$, which implies convergence in $L^1$ and in probability. Convergence in $L^1$ implies that the sequence $\{\hat{x}^i(k)\}$ is uniformly integrable (UI) for all $i$ \cite[Thm. 5.5.2]{durrett_probability10}. We claim that this implies that the sequence of FD weights $W(\hat{x}^i(k))$ computed in (\ref{eq:rbffd_weights}) is UI for each $i$. The matrix $\Phi$ in (\ref{eq:rbffd_weights}) is a bounded continuous function of $\hat{x}^i(k)$, which means that there exists a constant $K^\Phi_i \leq \infty$ for each $i$ such that $\|\Phi(\hat{x}^i(k))^{-T}\|_1 \leq K^\Phi_i$. Define $\alpha_i(k) :=  \hat{x}_i^i(k) - \sum_{j=1}^n \hat{x}^i_j(k) / n$. From (\ref{eq:rbffd_rhs}):
\begin{align*}
\|W&(\hat{x}^i(k))\|_1 \leq \left\|\begin{bmatrix}
  2\delta^2 e^{-\delta^2\|\alpha_1(k)\|_2^2}\alpha_1^T(k)\\
  \vdots\\
  2\delta^2 e^{-\delta^2\|\alpha_n(k)\|_2^2}\alpha_n^T(k)
\end{bmatrix}^T\right\|_1 \left\|\Phi(\hat{x}^i(k))^{-T}\right\|_1\\
&\leq 2\delta^2 K^\Phi_i\sum_{j=1}^n  e^{-\delta^2\|\alpha_j(k)\|_2^2} \|\alpha_j(k)\|_1\\
&\leq 2\delta^2 K^\Phi_i\sum_{j=1}^n \biggl\|\hat{x}^i_j(k) - \frac{1}{n} \sum_{l=1}^n \hat{x}^i_l(k)\biggr\|_1\\
&\leq 4 \delta^2 K^\Phi_i\sum_{j=1}^n \|\hat{x}^i_j(k)\|_1 = 4 \delta^2 K^\Phi_i\|\hat{x}^i(k)\|_1.
\end{align*}
By UI of $\{\hat{x}^i(k)\}$, for any $\epsilon > 0$, there exist $K_i \in [0,\infty)$ such that $\mathbb{E}\bigl[\|\hat{x}^i(k)\|_1 \indicator_{\{\|\hat{x}^i(k)\|_1 \geq K_i\}}\bigr]\leq\epsilon$ for all $k$. Then for all $i,k$:
\begin{align*}
\mathbb{E} &\bigl[\|W(\hat{x}^i(k))\|_1 \indicator_{\{\|W(\hat{x}^i(k))\|_1 \geq 4\delta^2 K^\Phi_i K_i\}} \bigr] \\
&\leq 4\delta^2 K^\Phi_i \mathbb{E} \biggl[ \|\hat{x}^i(k)\|_1 \indicator_{\{4\delta^2K_i^\Phi\|\hat{x}^i(k)\|_1 \geq 4\delta^2 K^\Phi_i K_i\}} \biggr]\leq 4\delta^2 K^\Phi_i \epsilon.
\end{align*}
Since $W(\hat{x}^i(k))$ is a continous function of $\hat{x}_i(k)$ by the continuous mapping theorem, $W(\hat{x}^i(k)) \overset{p}{\rightarrow} W(x^*), \; \forall i$. This, coupled with the uniform integrability of $\{W(\hat{x}^i(k))\}$ for all $i$ implies that $W(\hat{x}^i(k)) \overset{L^1}{\rightarrow} W(x^*), \; \forall i$. The signal measurements $z_i(\tau)$ in \eqref{eq:local_grad_estm} are independent of the estimates $W(\hat{x}^i(k))$ because the latter are based on the relative measurements in (\ref{eq:rel_mes}). Therefore,
\begin{align}
\mathbb{E} \hat{g}_{i}(k) &= \textstyle{\mathbb{E} \bigl[\mathbf{col}_i(W(\hat{x}^i(k))) \bigr] \frac{1}{k+1} \sum_{\tau=0}^k \mathbb{E} z_i(\tau)}\label{eq:grad_exp_convergence}\\
&=\mathbb{E} \bigl[\mathbf{col}_i(W(\hat{x}^i(k))) \bigr] h(x_i^*,y) \to \mathbf{col}_i(W(x^*)) h(x_i^*,y). \notag
\end{align}
Now, consider the behavior of the consensus filter in (\ref{eq:consensus_filter}) with $\mu_{i,k} = \hat{g}_i(k)$. Eliminating the state $q_{i,k}$ and writing the equations in matrix form gives:
\begin{center}$r_{k+1} = \left(I_{nd_x} - \beta (L \otimes I_{d_x})\right) r_k + (\mu_{k+1} - \mu_k),$\end{center}
where $L$ is the Laplacian of the communication graph $G$. Taking expectations above results is a deterministic linear time-invariant system, which was analyzed in \cite{Spanos_IFAC05}. In light of (\ref{eq:grad_exp_convergence}), Proposition 1 in \cite{Spanos_IFAC05} shows that for all $i$:
\[
\lim_{k\to\infty} \left(\mathbb{E}[r_{i,k}] - \frac{1}{n} \sum_{i=1}^n \mathbf{col}_i(W(x^*)) h(x_i^*,y) \right) = 0.
\]
Finally, the FD approximation in (\ref{eq:grad_approx}) shows that:
\[
\lim_{k\to\infty} \mathbb{E} r_{ik} = \frac{1}{n} W(x^*) \begin{pmatrix}
h(x_1^*,y)\\
\vdots\\
h(x_n^*,y)
\end{pmatrix} = \frac{1}{n} \biggl(g(m^*,y) + b\biggr), \; \forall i. \quad \blacksquare
\]

\section*{Appendix B: Proof of Thm. \ref{thm:dist_estm_L2}}
\noindent Define the following:
\begin{align*}
\omega_k &:= \begin{bmatrix} \omega_{1k}^T & \hdots & \omega_{nk}^T \end{bmatrix}^T \quad &\Omega_k&:= \begin{bmatrix} \Omega_{1k}^T & \hdots & \Omega_{nk}^T \end{bmatrix}^T\\
M_i &:= H_i^TE_i^{-1}H_i\quad &M&:=  \begin{bmatrix} M_1^T & \hdots & M_n^T \end{bmatrix}^T\\
\xi(k) &:=\rlap{$\begin{bmatrix} H_1E_1^{-T}\epsilon_1(k)^T & \hdots & H_nE_n^{-T}\epsilon_n(k)^T \end{bmatrix}^T$.}
\end{align*}
Then, (\ref{eq:gauss_dist_estm}) can be written in matrix form as follows:
\begin{equation}
\label{eq:wO_system}
\begin{aligned}
\omega_{k+1} &= \bigl(A \otimes I_{d_\theta}\bigr)\omega_k + M\theta^* + \xi(k),\\
\Omega_{k+1} &= \bigl(A \otimes I_{d_\theta}\bigr)\Omega_k + M,
\end{aligned}
\end{equation}
where $A =[a_{ij}]$, with $a_{ij} = 0$ if $j \notin \mathcal{N}_i \cup \{i\}$, is a stochastic matrix. The solutions of the linear systems are:
\begin{align*}
\omega_k &= \bigl(A \otimes I_{d_\theta}\bigr)^k \omega_0 + \sum_{\tau=0}^{k-1}\bigl(A \otimes I_{d_\theta}\bigr)^{k-1-\tau}\biggl(M \theta^* + \xi(\tau)\biggr),\\
\Omega_k &= \bigl(A \otimes I_{d_\theta}\bigr)^k \Omega_0 + \sum_{\tau=0}^{k-1}\bigl(A \otimes I_{d_\theta}\bigr)^{k-1-\tau} M.
\end{align*}
Looking at the $i$-th components again, we have:
\begin{align*}
\frac{\omega_{ik}}{k+1} &=\frac{1}{k+1} \sum_{j=1}^n \bigl[A^{k}\bigr]_{ij}\omega_{j0} +\\
&\quad\; \frac{1}{k+1} \sum_{\tau=0}^{k-1} \sum_{j=1}^n \bigl[A^{k-\tau-1}\bigr]_{ij}(M_j\theta^* + H_j^T E_j^{-1}\epsilon_{j}(\tau)),\\
\frac{\Omega_{ik}}{k+1} &= \frac{1}{k+1} \sum_{j=1}^n \bigl[A^k\bigr]_{ij}\Omega_{j0} +\frac{1}{k+1} \sum_{\tau=0}^{k-1} \sum_{j=1}^n \bigl[A^{k-\tau-1}\bigr]_{ij}M_{j}.
\end{align*}
Define the following to simplify the notation:
\begin{align}
b_{ik}&:= \textstyle{\frac{1}{k+1} \sum_{j=1}^n \bigl[A^k\bigr]_{ij}\omega_{j0}}, \quad &B_{ik}&:= \textstyle{\frac{1}{k+1} \sum_{j=1}^n \bigl[A^k\bigr]_{ij}\Omega_{j0}},\notag\\
c_{ik} &:= b_{ik} - B_{ik}\theta^*, \quad &C_{ik}& := \textstyle{\frac{1}{k+1}\Omega_{ik}},\notag\\
d_{it} &:=\rlap{$\textstyle{\frac{1}{k+1} \sum_{\tau=0}^{k-1} \sum_{j=1}^n \bigl[A^{k-\tau-1}\bigr]_{ij} H_j^T E_j^{-1}\epsilon_{j}(\tau)}$},\label{eq:shorthand-notation}\\
D_{ik} &:= \rlap{$\textstyle{\frac{1}{k+1} \sum_{\tau=0}^{k-1} \sum_{j=1}^n \bigl[A^{k-\tau-1}\bigr]_{ij}M_{j}}$}.\notag
\end{align}
With the shorthand notation:
\begin{equation}
\label{eq:wO_evolution}
  \frac{\omega_{ik}}{k+1} = b_{ik}+d_{ik}+D_{ik}\theta^*, \qquad C_{ik}=\frac{\Omega_{ik}}{k+1} =  B_{ik}+D_{ik},
\end{equation}
where $d_{ik}$ is the only random quantity. Its mean is zero because the measurement noise has zero mean, while its covariance is:
\begin{align}
  &\mathbb{E} [d_{ik} d_{ik}^T] =\frac{1}{(k+1)^2}\mathbb{E}\biggl[\biggl(\sum_{\tau=0}^{k-1} \sum_{j=1}^n \bigl[A^{k-\tau-1}\bigr]_{ij}H_{j}^T E_j^{-1} \epsilon_{j}(\tau)\biggr)\notag\\
  &\qquad \times \biggl(\sum_{s=0}^{k-1}\sum_{\eta=1}^n \bigl[A^{k-s-1}\bigr]_{i\eta}H_{\eta}^TE_\eta^{-1} \epsilon_{\eta}(s)\biggr)^T\biggr]\notag\\
  &\!=\!\frac{1}{(k+1)^2}\!\sum_{j=1}^n \sum_{\tau=0}^{k-1} \bigl[A^{k-\tau-1}\bigr]_{ij}^2H_j^T E_j^{-1} \mathbb{E}[\epsilon_{j}(\tau)\epsilon_{j}(\tau)^T]E_j^{-1}H_{j}\notag\\
  &= \frac{1}{(k+1)^2}\sum_{j=1}^n\sum_{\tau=0}^{k-1} \bigl[A^{k-\tau-1}\bigr]_{ij}^2M_j \preceq \frac{1}{k+1} D_{ik},\label{eq:covariance-dik}
\end{align}
where the second equality uses the fact that $\epsilon_j(\tau)$ and $\epsilon_\eta(s)$ are independent unless the indices coincide, i.e. $\mathbb{E}[\epsilon_j(\tau) \epsilon_\eta(s)^T] = \delta_{\tau s} \delta_{j \eta} E_j$. The L{\"o}wner ordering inequality in the last step uses that $0\leq\bigl[A^{k-\tau-1}\bigr]_{ij}\leq 1$ and $M_j \succeq 0$.

Since the communication graph $G$ is connected, $A$ corresponds to the transition matrix of an aperiodic irreducible Markov chain with a unique stationary distribution $\pi$ so that $A^k \to \pi \mathbf{1}^T$ with $\pi_j >0$. This implies that, as $k \to \infty$, the numerators of $b_{ik}$ and $B_{ik}$ remain bounded and therefore $b_{ik} \to 0$ and $B_{ik} \to 0$. Since Ces\'aro means preserve convergent sequences and their limits:
\[
\frac{1}{k+1} \sum_{\tau=0}^{k-1} \bigl[A^{k-\tau-1}\bigr]_{ij} \to \pi_j, \quad \forall i,
\]
which implies that $D_{ik} \to \sum_{j=1}^n \pi_j M_j$. The full-rank assumption on $\begin{bmatrix} H_1^T & \hdots & H_n^T \end{bmatrix}^T$ and $\pi_j>0$ guarantee that $\sum_{j=1}^n \pi_j M_j$ is positive definite. Finally, consider the mean squared error:
\begin{align*}
&\mathbb{E}\bigl[(\hat{\theta}_i(k) - \theta^*)^T(\hat{\theta}_i(k) - \theta^*) \bigr]\\
&= \mathbb{E} \biggl\| \biggl(\frac{\Omega_{ik}}{k+1}\biggr)^{-1}\frac{\omega_{ik}}{k+1} -  \biggl(\frac{\Omega_{ik}}{k+1}\biggr)^{-1}\biggl(\frac{\Omega_{ik}}{k+1}\biggr)\theta^* \biggr\|_2^2\\
&= \mathbb{E} \bigl\| C_{ik}^{-1} \bigl(b_{ik} +  d_{ik} + D_{ik}\theta^* - (B_{ik} + D_{ik})\theta^* \bigr)  \bigr\|_2^2\\
&= \mathbb{E} \|C_{ik}^{-1}(c_{ik} +d_{ik})\|_2^2\\
&= \mathbb{E}\biggl[c_{ik}^T C_{ik}^{-T} C_{ik}^{-1} c_{ik} + 2c_{ik}^TC_{ik}^{-T} C_{ik}^{-1}d_{ik} + d_{ik}^TC_{ik}^{-T} C_{ik}^{-1}d_{ik}\biggr]\\
&\longeq{(a)}{} c_{ik}^T C_{ik}^{-T} C_{ik}^{-1} c_{ik} + \tr( C_{ik}^{-1} \mathbb{E} [d_{ik} d_{ik}^T] C_{ik}^{-T})\\
&\longineq{(b)}{}{\leq} c_{ik}^T C_{ik}^{-T}C_{ik}^{-1}c_{ik} +\frac{1}{k+1}\tr(C_{ik}^{-1}D_{ik} C_{ik}^{-T}) \to 0,
\end{align*}
where $(a)$ holds because the first term is deterministic, while the cross term contains $\mathbb{E}[d_{ik}] = 0$. Inequality $(b)$ follows from \eqref{eq:covariance-dik}. In the final step, as shown before $C_{ik}^{-1} \to \bigl(\sum_{j=1}^n \pi_j M_j \bigr)^{-1}$ and $D_{ik} \to \sum_{j=1}^n \pi_j M_j \succ 0$ remain bounded, while $c_{ik} \to 0$ and $1/(k+1) \to 0$.\hfill$\blacksquare$

\bibliographystyle{asmems4}

\bibliography{bib/sss}

\end{document}